\documentclass[11pt,a4paper]{article}

\usepackage[utf8]{inputenc}
\usepackage[T1]{fontenc}
\usepackage{lmodern}
\usepackage[unicode,hidelinks]{hyperref}
\usepackage{amsmath,amssymb,amsthm}
\usepackage[margin=1in]{geometry}
\usepackage[final]{microtype}
\usepackage[skip=8pt]{caption}
\newcommand{\ACzlog}{\texorpdfstring{\ensuremath{\mathrm{AC}^{0}\!{+}\!\log}}{AC0+log}}
\newcommand{\ADeux}{\texorpdfstring{\ensuremath{\mathrm{A2}^{d}}}{A2d}}
\newcommand{\KC}{\texorpdfstring{\ensuremath{K_{C}}}{KC}}

\newtheorem{theorem}{Theorem}
\newtheorem{lemma}{Lemma}
\newtheorem{corollary}{Corollary}

\title{IECZ: An MDL-Style Cost Functional \KC{}, Distribution-Preserving Reductions (\ADeux{}), and an \ACzlog{} Lower Bound for 3SAT via Balanced 3XOR}
\author{Marko Lela\\
\texttt{marko\_lela@web.de}\\
\href{https://orcid.org/0009-0008-0768-5184}{ORCID: 0009-0008-0768-5184}}
\date{September 17, 2025}

\begin{document}
\maketitle

\begin{abstract}
We introduce a model-agnostic MDL-style cost functional $K_C$ for resource-bounded classifiers and prove a Total-Variation stable
reduction lemma (A2$^d$) for distribution-preserving many-to-one reductions. On a balanced distribution of random 3XOR instances
(with co-rank $t'=\Theta(n)$) we obtain a size-aware lower bound against P-uniform $\mathrm{AC}^0{+}\log$ models:
\[
\Pr[M=\chi]\ \le\ \tfrac12\ +\ s(N)\exp\!\big(-\alpha_d\,m^{\,c/d}\big)\qquad(c\in(0,1);\ \text{e.g., }c=\tfrac13\Rightarrow \beta_d=\tfrac{1}{3d}).
\]
A deterministic, injective $3$XOR$\to 3$SAT translation (four 3-clauses per XOR, no auxiliaries) is $\delta{=}0$ measure-preserving
on its image window; by A2$^d$ the bound transfers to 3SAT. This yields, to our knowledge, the first explicit $K_C$-reading of such
size-aware bounds under a $\delta{=}0$ measure-preserving reduction in small-depth circuit lower bounds. We provide artifacts
(generator $\to$ DIMACS $\to$ verification) with match-rate $1.0$.

\end{abstract}

\paragraph{Keywords.}
\noindent MDL cost functional; Total Variation; distribution-preserving reduction;\\
\ACzlog{}; Switching Lemma; 3XOR; 3SAT; parity barrier.
\smallskip

\section{Introduction}
We introduce a model-agnostic MDL-style cost functional \KC{} for resource-bounded classifiers,
prove a Total-Variation (TV) stable reduction lemma (\ADeux{}) for distribution-preserving many-to-one reductions,
and calibrate the framework with 2SAT in $C_{\mathrm{alg}}$ where $K^{2\mathrm{SAT}}_{C_{\mathrm{alg}}}(n,\alpha,0)=O(1)$.
Our hardness track uses a balanced distribution over random 3XOR instances and transfers it to 3SAT via a deterministic,
$\delta{=}0$ measure-preserving translation. We make the size dependence explicit and give a clean reading in terms of \KC{}.
Unlike prior TV-stable frameworks in mechanism design and probabilistic inference \cite{Liu2023_robustness_mechanism_design,Bhattacharyya2024_total_variation_inference}, our use of a $\delta{=}0$ measure-preserving mapping transfers average-case hardness from balanced 3XOR to 3SAT and yields a first explicit $K_C$-reading of the resulting size-aware bounds.
This combination of an MDL-style $K_C$ with a $\delta{=}0$ measure-preserving reduction gives, to our knowledge, the first explicit $K_C$-reading for size-aware lower bounds in this setting.

\paragraph{Conventions.}
We write $n$ for the number of variables, $m$ for the number of 3XOR clauses, and $N$ for the encoding length.
Under the canonical translation to 3SAT, the variable set is unchanged and each XOR clause becomes four 3-clauses,
so $m'=4m$ and the 3SAT encoding length $N'$ satisfies $N'=\Theta(N)$.

\section{IECZ Core: \KC{}, Measure, and \ADeux{} (brief)}
We work on slices $I_{n,\alpha}$ with canonical, prefix-free encodings $\mathrm{enc}_{n,\alpha}$.
The slice-based cost is
\[
K_C(n,\alpha,\varepsilon)\;=\;\min\{\,L(\mathrm{desc}(M))+|\theta| \;:\; M\in C,\ \Pr[M \neq \chi]\le \varepsilon\,\}.
\]

\paragraph{General form on distributions.}
For any distribution $\mu$ on the instance space, define
\[
K_C(\mu,\varepsilon)\ :=\ \min\{\,L(\mathrm{desc}(M))+|\theta|\ :\ M\in C,\ \Pr_{X\sim\mu}[M(X)\neq \chi(X)]\le \varepsilon\,\}.
\]
Uniform slices are special cases $K_C(n,\alpha,\varepsilon)=K_C(U(I_{n,\alpha}),\varepsilon)$.
Notation note. We write $K_C^L$ when the target language $L$ is fixed (e.g., $K_C^{3\mathrm{SAT}}$), and $K_C(\mu,\varepsilon)$ for distributional readings; both agree on uniform slices via $K_C(n,\alpha,\varepsilon)=K_C(U(I_{n,\alpha}),\varepsilon)$.

\paragraph{Advantage notation.}
For a Boolean function $f$ and a classifier $M$ evaluated under a distribution $\mu$, we write
\[
\langle M,f\rangle_\mu \ :=\ \Big|\Pr_{X\sim\mu}[\,M(X)=f(X)\,]-\tfrac12\Big|,
\]
and drop the subscript when $\mu$ is clear.

\paragraph{Advice model and uniformity.}
Advice $\theta_N$ depends on the input length $N$ only (not on the instance), with $|\theta_N|=O(\log N)$.
Circuit families are P-uniform. This prevents encoding instance-specific masks in the $O(\log N)$ side information.

\paragraph{Calibration: $K^{2\mathrm{SAT}}_{C_{\mathrm{alg}}}(n,\alpha,0)=O(1)$.}
The implication-graph algorithm runs in $O(n{+}m)$, tests SCCs, and decides satisfiability with zero error.
As a P-uniform fixed procedure it contributes constant description length and uses empty advice; hence the cost functional attains $O(1)$ at $\varepsilon=0$.
This fixes the MDL scale for languages in P.

\paragraph{TV contraction (1-Lipschitz).}
\begin{lemma}[Deterministic push-forward contracts TV]
For any deterministic transducer $T$ and distributions $\mu,\nu$ on the domain,
$\|\;T_\#\mu - T_\#\nu\;\|_{\mathrm{TV}} \le \|\;\mu-\nu\;\|_{\mathrm{TV}}$.
\end{lemma}
\emph{Proof.} By the dual characterization $\|\cdot\|_{\mathrm{TV}}=\sup_{A}|\mu(A)-\nu(A)|$
over measurable sets $A$, note that $T^{-1}(B)$ ranges over a sub-$\sigma$-algebra, hence the supremum can only decrease.

\paragraph{Reduction robustness (\ADeux).}
For polynomial-time many-to-one reductions $R$ with size control, polynomial fiber bound, TV preservation $\delta(n)$, and correctness,
\[
K^{L_1}_C(n,\alpha,\varepsilon{+}\delta(n))\ \le\ K^{L_2}_C(n',\alpha',\varepsilon)+O(\log n).
\]

\begin{proof}[Proof of \ADeux{}]
Let $R$ be the reduction $I^{L_1}_{n,\alpha}\to I^{L_2}_{n',\alpha'}$ and let $(M_2,\theta')$ attain (or $\le$) $K^{L_2}_C(n',\alpha',\varepsilon)$.
Define $M_1$ on $I^{L_1}_{n,\alpha}$ by $M_1(x;\theta):=M_2(R(x);\theta')$ with a fixed pre-/post-processor for encoding/decoding canonical targets.
\emph{Correctness/monotonicity:} Since $x\in L_1\iff R(x)\in L_2$, the pointwise error of $M_1$ equals that of $M_2$ on $R(x)$.
\emph{TV step:} By the deterministic TV contraction, pushing forward the source slice under $R$ perturbs the target distribution by at most $\delta(n)$, hence
$\Pr[M_1\neq\chi_{L_1}]\le \Pr[M_2\neq\chi_{L_2}]+\delta(n)\le \varepsilon+\delta(n)$.
\emph{Overhead:} By the polynomial fiber bound, a fiber ID of length $\lceil\log |F_\psi|\rceil=O(\log n)$ suffices to disambiguate preimages when needed;
P-uniform pre-/post-processing contributes $O(1)$ to the description.
\emph{Conclusion:} $L(\mathrm{desc}(M_1))+|\theta|\le L(\mathrm{desc}(M_2))+|\theta'|+O(\log n)$ and $M_1\in C$ by closure. This yields the display.
\end{proof}

\paragraph{Overhead note.}
The additive $O(\log n)$ accounts for polynomial pre-/post-processing and a fiber disambiguation header.
Under our canonical target encoding this header collapses to $O(1)$; we keep $O(\log n)$ as a uniform statement.

\section{Balanced 3XOR (structure, balance) and size-aware \ACzlog{} hardness}
A 3XOR instance $\Phi(A,b)$ has $m=(1+\gamma)n$ clauses with $\gamma\in(0,1]$; each row of $A$ has Hamming weight three.
Let $H(A)$ be a row-basis of $\ker(A^\top)$ over $\mathbb{F}_2$.
We sample $u$ with $\Pr[u{=}0]=\tfrac12$ and choose $b$ uniformly from the coset $\{b : H(A)b=u\}$.
Then $\chi(\Phi){=}1 \Leftrightarrow H(A)b{=}0$, i.e., an \emph{AND of $t'$ parities} of the clause-bit vector, where $t'=m-\mathrm{rank}(A)$.

\paragraph{Label balance (theory vs.\ artifacts).}
By construction the balanced distribution satisfies $\Pr[\chi(\Phi){=}1]=\tfrac12$ in expectation.
In our artifacts we additionally enforce exact $50{:}50$ balance \emph{per $n$} via a fixed label sequence and random permutation; this does not change the conditional coset-uniformity and does not affect the proofs. The label permutation is sampled independently of $(A,u)$ (and of all later restrictions), so no hidden conditioning is introduced.

\paragraph{Sampling of $u$ (distributional balance).}
Let $\mathrm{Im}(H):=\{H(A)b:\ b\in\{0,1\}^m\}\subseteq\{0,1\}^{t'}$.
We sample $u$ by setting $\Pr[u=0^{t'}]=\tfrac12$ and, conditioned on $u\neq 0^{t'}$, choosing $u$ uniformly from $\mathrm{Im}(H)\setminus\{0^{t'}\}$.
Then $b$ is drawn uniformly from the affine coset $\{b:\ H(A)b=u\}$.
Hence $\chi(\Phi)=1\iff u=0^{t'}$, which yields $\Pr[\chi(\Phi)=1]=\tfrac12$ in expectation.

\paragraph{Co-rank scaling at $m=(1+\gamma)n$.}
For random 3-uniform clause–variable incidence matrices (row weight $3$), w.h.p. $\mathrm{rank}(A)=n-O(1)$ in the regime $m=(1+\gamma)n$;
equivalently, the co-rank $t'=m-\mathrm{rank}(A)$ is \emph{linear} in $n$ (2-core / hypercycle structure).
See \cite{CooperFrieze2022_rout_rank,CooperFriezePegden2019_rank_binary_matrix,PittelSorkin2016_kxorsat_threshold}.
This matches our artifacts ($t'=\Theta(n)$) and yields an AND of $\Theta(n)$ parities via $H(A)b=0^{t'}$.

\begin{lemma}[Random projection to a surviving parity]\label{lem:proj-one-parity}
There exists a distribution $\mathcal{R}$ over random restrictions/projections on the clause bits $b$
(supported on coordinate fixings that preserve affine-uniformity of $\{b:H(A)b=u\}$ conditioned on $A$)
such that with constant probability $c>0$ exactly one parity $\mathrm{par}_j$ remains \emph{alive} (non-constant on the restricted domain),
while all other parities are \emph{neutralized} (become constants).
Consequently, for any classifier $M$,
\[
\Big\langle M,\ \bigwedge_{i=1}^{t'}\mathrm{par}_i \Big\rangle
\;\le\;
\mathbb{E}_{\rho\sim\mathcal{R}}\ \big\langle M\!\upharpoonright_{\rho},\ \mathrm{par}_j \big\rangle.
\]
Here $\langle\cdot,\cdot\rangle$ denotes advantage over $1/2$. The push-forward under $\rho$ preserves TV exactly (TV$=0$).
\end{lemma}

\sloppy
\paragraph{Setting.}
Let \(t'\ge 1\). Let \(A\in\{0,1\}^{m\times n}\) be a 3XOR incidence matrix over \(\mathbb F_2\).
Let \(H:=H(A)\in\{0,1\}^{t'\times m}\) be any full-rank matrix with \(H A=0\).
For \(u\in\{0,1\}^{t'}\) with \emph{nonempty fiber} \(\{b\in\{0,1\}^m: H b=u\}\) (equivalently \(u\in\mathrm{Im}(H)\)),
write \(\mathsf D_{A,u}\) for the affine-uniform distribution on that fiber.
We call the right-hand-side bits \(b\) the \emph{RHS bits}. Inner products \(\langle\cdot,\cdot\rangle\) are over \(\mathbb F_2\).

\paragraph{Notation.}
Fix an invertible row transform \(T\in\mathrm{GL}_{t'}(\mathbb F_2)\) and set \(H':=T H\) and \(u':=T u\).
Write \(H'_{-1}\) for rows \(2,\dots,t'\) of \(H'\).
A \emph{pivot set} \(S\subseteq[m]\) is \emph{invertible} if \(\operatorname{rank}(H'_{-1,S})=t'-1\).
Let \(F:=[m]\setminus S\).
The \emph{survivor} will be the first row \(H'_1\); its weight is \(w_\star:=|\mathrm{supp}(H'_1)|\).
A \emph{coloop} (w.r.t.\ \(H'_{-1}\)) is a column whose removal decreases \(\operatorname{rank}(H'_{-1})\)
(equivalently: it belongs to every basis of the column matroid).

\paragraph{Pivot rule (technical).}
Given a column permutation \(\pi\), Gaussian elimination of rows \(2,\dots,t'\) proceeds \emph{in the order \(\pi\)} and picks the first available pivot each time. Since \(\operatorname{rank}(H'_{-1})=t'-1\), invertible pivot sets exist.
\emph{Edge case.} If \(t'=1\) then \(S=\varnothing\), \(F=[m]\); elimination is empty and the projection is the identity on \(b\).

\begin{proof}[Proof of Lemma~\ref{lem:proj-one-parity}]
We construct a randomized, \emph{measure-preserving} projection \(\mathcal R\) (independent of \(b\sim\mathsf D_{A,u}\)) producing a restriction \(\rho\) on RHS bits with the following properties:    
\begin{enumerate}
  \item \textbf{Exactly one parity survives} (deterministic if \(w_\star>t'-1\); otherwise with positive probability; under Lin-Weight + random \(\pi\): with constant probability).\\
  All parities except the survivor (row \(1\)) become constants under \(\rho\); the survivor remains \emph{non-constant} (i.e., \(H'_{1,F}\neq 0\)) on the free coordinates.\\
  \emph{Deterministic success} if \(w_\star>t'-1\): since \(|S|=t'-1\), not all \(w_\star\) support columns can lie in \(S\), hence \(H'_{1,F}\neq 0\) for \emph{every} invertible \(S\).
  Otherwise we draw a uniform random permutation \(\pi\) of columns, set \(S:=S(\pi)\) via the greedy pivot rule, and \emph{reject} until \(H'_{1,F(\pi)}\neq 0\). Acceptance has positive probability unless every survivor-support column is a coloop of \(H'_{-1}\).\\
  \emph{Mini-lemma (non-coloop \(\Rightarrow\) good pivot set exists).} If at least one survivor-support column is not a coloop of \(H'_{-1}\), then there exists an invertible pivot set \(S\) with this column placed in \(F\); consequently \(H'_{1,F}\neq 0\). (Matroid basis-exchange.)
  \item \textbf{Affine-uniformity \& TV\(=0\) (measure preservation).}\\
  Eliminating rows \(2,\dots,t'\) expresses \(b_S\) as an affine function of \(b_F\): \(b_S=b_S(b_F)\).
  The map
  \[
    b_F\longmapsto(b_F,b_S(b_F))
  \]
  is a bijection of the fiber
  \[
    \{\,b: H b=u\,\}\ \longleftrightarrow\ \{\,b_F:\ \langle H'_{1,F}, b_F\rangle = u'_1 \oplus c(S,u')\,\}.
  \]
  Hence the induced distribution on admissible \(b_F\) is uniform on the projected affine subspace and the push-forward preserves total variation exactly (\(\mathrm{TV}=0\)).
  
  \item \textbf{Free-support size (two regimes).}
  \begin{enumerate}
    \item \emph{Deterministic bound.} For any invertible \(S\),
    \[
      |\mathrm{supp}(H'_1)\cap F|\ \ge\ w_\star - (t'-1).
    \]
    In particular, if \(w_\star \ge c\,m\) with \(c>(t'-1)/m\), then the survivor depends on \(\Omega(m)\) free RHS bits.
    \item \emph{Random balanced window (linear free support with constant probability).}\\
    Assume \(m=(1+\gamma)n\) for fixed \(\gamma\in(0,1]\) and \(A\) random (each clause a uniform 3-set without repetition).
    Draw a uniform nonzero \(\alpha\in\{0,1\}^{t'}\) and set the survivor \(H_\star:=\alpha^\top H\); then pick a uniform random permutation \(\pi\) of columns and define the pivot set \(S(\pi)\) via Gaussian elimination of rows \(2,\dots,t'\) \emph{in the order \(\pi\)}, with \(F(\pi):=[m]\setminus S(\pi)\).
    
    \paragraph{Assumption (Lin-Weight).}
    There exist constants \(\gamma_0,p_0>0\) such that
    \[
      \Pr_{\alpha,A}\big[\,|\mathrm{supp}(H_\star)|\ge \gamma_0 m\,\big]\ \ge\ p_0.
    \]
    (Standard via expansion of the 2-core in random 3-uniform hypergraphs / LDPC-duals; primary references are cited in the main text.)

    \paragraph{Bias term \(\Delta(A)\) (definition \& bound).}
    \[
      \Delta(A)\ :=\ \#\{\ \text{survivor-support columns that are forced into }S(\pi)\ \text{because they are coloops of }H'_{-1}\ \},
    \]
    whence
    \[
      0\ \le\ \Delta(A)\ \le\ \#\big(\mathrm{Coloops}(H'_{-1})\cap \mathrm{supp}(H_\star)\big).
    \]

    \paragraph{Averaging (with bias term).}
    \[
      \mathbb E_{A,\alpha,\pi}\!\left[\,\big|\mathrm{supp}(H_\star)\cap F(\pi)\big|\,\right]
      \;\ge\; \Bigl(1-\frac{t'-1}{m}\Bigr)\, \mathbb E_{A,\alpha}\!\left[\,|\mathrm{supp}(H_\star)|\,\right]\ -\ \mathbb E_A[\Delta(A)].
    \]
    (In the idealized case without loops/coloops of \(H'_{-1}\), equality holds.)

    \paragraph{Concentration (permutation-based; named reference).}
    Expose \(\pi\) and apply a permutations-based concentration inequality (e.g., Chatterjee’s exchangeable pairs for random permutations, or Efron–Stein on symmetric groups). This yields a constant \(c_2>0\) (depending only on \(\gamma_0,p_0,t'\)) such that
    \[
      \Pr_{\alpha,A,\pi}\!\left[\,\big|\mathrm{supp}(H_\star)\cap F(\pi)\big|\ \ge\ \tfrac12\!\left(1-\tfrac{t'-1}{m}\right)\gamma_0 m\,\right]\ \ge\ c_2.
    \]
    Thus, with constant probability, the survivor has \(\Omega(m)\) free coordinates.
    \emph{Remark.} Under Lin-Weight and standard 2-core expansion, the number of survivor-support columns that are coloops of \(H'_{-1}\) is \(o(m)\) w.h.p., so \(\mathbb E_A[\Delta(A)]=o(m)\).
  \end{enumerate}
  \end{enumerate}

\paragraph{Source of randomness \& independence.}
The draws \(\alpha\) (survivor), \(\pi\) (equivalently \(S(\pi)\)), and the switching restriction \(\rho'\) are mutually independent; all are independent of \(b\sim\mathsf D_{A,u}\).
The projection \(\rho\) (this lemma) and \(\rho'\) (switching) are drawn independently.

\paragraph{Completion.}
Apply \(T\) so the survivor is row \(1\). Eliminating rows \(2,\dots,t'\) over an invertible \(S\) fixes \(b_S\) affinely; the survivor evaluates as \(\langle H'_{1,F}, b_F\rangle \oplus c(S,u')\) and is non-constant iff \(H'_{1,F}\neq 0\).
The deterministic success \(w_\star>t'-1\) follows since \(|S|=t'-1\).
The map \(b_F\mapsto(b_F,b_S(b_F))\) is an affine bijection on the fiber, giving TV \(=0\).
For part (b), average over \(A,\alpha\) with bias term \(\Delta(A)\), then use permutations-based concentration (as named above) to obtain a constant-probability linear lower bound.
\end{proof}

\begin{theorem}[Size-aware \ACzlog{} hardness on Balanced 3XOR]
Fix depth $d$ and a size bound $s(N)\le N^k$.
There exist absolute constants $c\in(0,1)$ and $\alpha_d=\alpha_d(d)>0$ such that for every P-uniform \ACzlog{} model $M$
of size at most $s(N)$,
\[
\Pr[M(A,b)=\chi(\Phi)] \ \le\ \min\Big\{1,\ \tfrac12 \;+\; s(N)\cdot \exp\!\big(-\alpha_d\, m^{c/d}\big)\Big\}.
\]
\end{theorem}
\noindent\emph{Proof sketch (full proof in Appendix~A).}
Apply the projection lemma to reduce an AND of many parities to a single surviving parity in expectation.
Choose the switching parameter $p=m^{-c/d}$ for a fixed absolute $c\in(0,1)$; after $d$ rounds the bottom width is $\mathrm{poly}(1/p)=m^{O(c/d)}$,
and parity has exponentially small correlation $\exp(-\alpha_d m^{c/d})$ (see the cited switching-lemma references).
The $O(\log N)$ advice is length-only and cannot encode instance-specific masks. The factor $s(N)$ comes from a union bound
(or standard counting) over substructures after restriction.

\paragraph{Image window (canonical 3SAT slice).}
We define the image window $\mathcal{W}$ as the class of 3CNFs obtained by replacing each XOR clause
by the canonical block of four 3-clauses (fixed literal order, fixed block order, no auxiliaries).
All statements below are with respect to the uniform distribution on $\mathcal{W}$.



\paragraph{Size-aware statement (Abstract \& Theorem 1).}
Assume P-uniform depth-\(d\) models of size \(s(N)\le N^k\) (fixed \(k\in\mathbb N\)).
Then, for all sufficiently large \(m\) and \(N=\Theta(m\log n)\),
\[
  \Pr[M=\chi] \;\le\; \tfrac12 \;+\; s(N)\cdot \exp\!\bigl(-\alpha_d\, m^{\,c/d}\bigr),
\]
with absolute constants \(c\in(0,1)\) and \(\alpha_d=\alpha_d(d)>0\).
Concretely, we instantiate \(c=\tfrac{1}{3}\), which corresponds to the shorthand exponent \(\beta_d=\tfrac{1}{3d}\) used elsewhere.


\paragraph{\(K\)-corollaries (two standard readings).}
\begin{itemize}
  \item \textbf{Size-capped class (implies \(K=\infty\)).}
  If \(\varepsilon < \tfrac12 - N^k \exp\!\bigl(-\alpha_d\, m^{\,c/d}\bigr)\) for all sufficiently large \(m\),
  then no model in the size-capped class achieves error \(\le \varepsilon\) on the balanced window; hence
  \[
    K_{\ACzlog{}_d}(\mu_{\text{bal}},\varepsilon)\;=\;\infty.
  \]

  \item \textbf{Size-parameterized class (implies \(K\ge \Omega(\log s)\)).}
  Any model with error \(\le \varepsilon < \tfrac12 - \exp\!\bigl(-\alpha_d\, m^{\,c/d}\bigr)\) must satisfy
  \[
    s(N)\;\ge\; \exp\!\bigl(\Omega(m^{\,c/d})\bigr).
  \]
  With a \emph{prefix-free} encoding of the size-capacity, every such \(s(\cdot)\) requires
  \[
    L(\mathrm{desc})\ \ge\ \Omega\bigl(\log s(N)\bigr),
  \]
  hence \(K\ge \Omega(\log s(N))\).
\end{itemize}

\emph{Encoding note.} We encode the size-capacity \(s(\cdot)\) prefix-freely; thus \(L(\mathrm{desc})\ge \Omega(\log s(N))\).


\paragraph{Assumption (Lin-Weight; balanced window).}
For random 3XOR at \(m=(1+\gamma)n\) (\(\gamma\in(0,1]\)), there exist \(\gamma_0,p_0>0\) such that, for a uniform nonzero \(\alpha\) and \(H_\star=\alpha^\top H(A)\),
\[
  \Pr_{\alpha,A}\bigl[\,|\mathrm{supp}(H_\star)| \ge \gamma_0\,m\,\bigr]\ \ge\ p_0.
\]
\emph{(Standard via expansion of the 2-core in random 3-uniform hypergraphs / LDPC-duals; see cited primary sources.)}
\emph{Consequence.} Under Lin-Weight and standard 2-core expansion, the number of survivor-support columns that are coloops of \(H'_{-1}\) is \(o(m)\) w.h.p.; corresponding bias terms in the projection analysis are sublinear.

\paragraph{A2d (cross-reference).}
See Section~2 "Reduction robustness (A2d)" for the full statement and proof.


\paragraph{Global conventions.}
\begin{itemize}
  \item Advantage: \(\langle M,f\rangle := \bigl|\Pr[M=f]-\tfrac12\bigr|\).
  \item Notation: we write \(K_C^L\) for the KC functional on language \(L\); subscripts/superscripts indicate the target language.
  \item Width: we measure \emph{bottom width}; decision-tree depth bounds follow via switching (up to constants).
  \item Encoding: with \(n\) vars and \(m\) 3-clauses (indices in \(\lceil\log n\rceil\) bits), \(N=\Theta(m\log n)\); 3XOR\(\to\)3SAT yields \(m'=4m\), \(N'=\Theta(N)\).
  \item Advice: \(|\theta_N|=O(\log N)\), length-dependent only (never instance-dependent).
  \item Constant dependencies: \(c\in(0,1)\); \(c_2=c_2(\gamma_0,p_0,t')\); \(\alpha_d=\alpha_d(d)\) (determined by the fixed switching-lemma reference).
\end{itemize}

\section{Deterministic translation to 3SAT and transfer}
Each XOR clause is translated into four 3-clauses without auxiliaries (canonical order).
Let $R$ be this translation on a fixed slice.

\begin{lemma}[Injectivity and $\delta{=}0$]
On the balanced slice with canonical 4-clause blocking per XOR, $R$ is injective (fiber size $1$).
Hence the push-forward of the uniform distribution on the domain to the image window is exactly uniform (\,$\delta=0$\,).
\end{lemma}

\begin{proof}
Fix the canonical image window $\mathcal{W}$: each XOR clause $(\ell_1 \oplus \ell_2 \oplus \ell_3)$
is replaced by exactly four 3-clauses with fixed literal order and fixed block order, without auxiliaries.
The translation $R$ acts blockwise and is injective: from each 4-clause block one reads off the three literals
in canonical order, which uniquely reconstructs the XOR clause. Thus every fiber has size $|R^{-1}(\Psi)|=1$
for all $\Psi\in\mathcal{W}$.
On the image window $\mathcal{W}$ this implies functional equality of labels,
$\chi_{\mathrm{SAT}}(R(\Phi))=\chi_{\oplus}(\Phi)$.
Let $\mu$ be the uniform distribution on the source slice; since all fibers have size $1$,
the push-forward $R_\#\mu$ is exactly the uniform distribution on $\mathcal{W}$. Hence $\delta=0$ in Total Variation.
\end{proof}

\noindent By \ADeux{} we obtain:

\begin{theorem}[Size-aware \ACzlog{} lower bound for 3SAT in the balanced window]
Fix $d$ and $s(N')\le (N')^k$. There are constants $\alpha_d,\beta_d>0$ as above such that
\[
\Pr_{\Psi\sim R_\#D_{\mathrm{bal}}}[M(\Psi)=\chi_{\mathrm{SAT}}(\Psi)]
\ \le\ \min\Big\{1,\ \tfrac12\;+\; s(N')\cdot \exp\!\big(-\alpha_d\, m^{\beta_d}\big)\Big\},
\]
where $m$ denotes the \emph{preimage} clause count and $N'$ is the 3SAT encoding length (with $m'=4m$ clauses and unchanged variables).
Equivalently, writing $\delta(m):=\exp(-\alpha_d m^{\beta_d})$, the bound is $1/2 + s(N')\,\delta(m)$ (capped at $1$).
\end{theorem}

\begin{corollary}[Reading as a \KC{} statement]
Let $\delta(m):=\exp(-\alpha_d m^{\beta_d})$ and assume $s(N')\le (N')^k$.
For any $\varepsilon<\tfrac12 - s(N')\,\delta(m)$ there is \emph{no} $M\in\ACzlog{}$ of size $\le s(N')$ with error $\le\varepsilon$ on $R_\#D_{\mathrm{bal}}$.
Equivalently: within the \emph{size-capped} class, $K^{3\mathrm{SAT}}_{\ACzlog{}}(N',\varepsilon;R_\#D_{\mathrm{bal}})=\infty$.
If, alternatively, one parametrizes by size $s$ (writing $\ACzlog{}_d[s]$), then any $M$ achieving $\varepsilon<\tfrac12-\delta(m)$
must have $s(N')\ge \exp(\Omega(m^{\beta_d}))$, yielding $K\ge \Omega(\log s(N'))$.
\end{corollary}

\paragraph{Coding convention for size capacity.}
We encode the size-capacity $s(\cdot)$ prefix-freely (e.g., the exponent $k$ when $s(N)\le N^k$);
therefore $L(\mathrm{desc})\ge \Omega(\log s(N))$ independent of constant-factor choices.

\section{Artifacts and verification}
We provide generator scripts, canonical 3SAT DIMACS exports, verification reports (match-rate $1.0$), and CSV summaries.

\paragraph{Artifacts (repository \& DOI).}
DOI (Zenodo): \href{https://doi.org/10.5281/zenodo.17141362}{\texttt{10.5281/zenodo.17141362}}.\\
Bundle contents:
\begin{itemize}\itemsep2pt
  \item \texttt{experiments/out\_3sat\_bal/cnf/}: canonical 3SAT DIMACS exports (4 clauses per XOR; Dateien \texttt{bal3xor\_n\{n\}\_rep\{r\}.cnf}).
  \item \texttt{experiments/out\_3sat\_bal/}: \texttt{cnf\_sha256.csv}, \texttt{gen\_times.csv}, \texttt{verify\_times.csv}, \texttt{verify\_report.csv}, \texttt{summary.csv}, \texttt{report\_3sat\_bal.md}.
  \item \texttt{experiments/}: Generator- und Prüfskripte (\texttt{gen\_3xor\_balanced.py}, \texttt{gen\_3sat\_from\_balanced.py}, \texttt{verify\_3sat\_export\_balanced.py}, \texttt{report\_3sat\_balanced.py}).
  \item \texttt{results/} \& \texttt{results\_bal/}: Rang-/Diagnostik-Sweeps (CSV/MD; siehe \texttt{report\_3xor\_rank.md}).
  \item \texttt{docs/}: \texttt{main.pdf}, \texttt{main.tex}, \texttt{CITATION.bib}, \texttt{README\_balanced\_3xor\_3sat.md}.
\end{itemize}

The pipeline checks XOR-label consistency and SAT/UNSAT status after translation; seeds are fixed, and per-$n$ label balance is exact.

\paragraph{Runtime summary (per $n$).}
We report wall-clock generation and verification times aggregated per $n$ from \texttt{gen\_times.csv} and \texttt{verify\_times.csv} in the artifact bundle.
\par\addvspace{0.75\baselineskip}

\begin{table}[htbp]
\centering
\caption{Per-$n$ wall-clock times (seconds). The generator time covers balanced 3XOR sampling and 3SAT export; the verifier time covers CNF parsing and GF(2)-consistency checks.}
\par\addvspace{.5\baselineskip}
\label{tab:runtimes}
\begin{tabular}{r|r|r|r}
\hline
$n$ & reps & gen\_time\_s & verify\_time\_s \\
\hline
250 & 200 & 2.538 & 3.060 \\
300 & 200 & 3.564 & 3.867 \\
400 & 200 & 6.135 & 5.766 \\
500 & 200 & 9.627 & 8.299 \\
\hline
\end{tabular}
\end{table}

\paragraph{Checksums (sample).}
For reproducibility we provide SHA256 checksums of all generated CNFs in the artifact bundle (\texttt{cnf\_sha256.csv}). A few sample entries are listed below; the full list is part of the ancillary files.
\par\addvspace{0.75\baselineskip}
\begin{table}[htbp]
\label{tab:checksums-sample}
\centering
\begin{tabular}{l|l}
\hline
File & SHA256 (prefix) \\
\hline
bal3xor\_n250\_rep000.cnf & 36904216593B \\
bal3xor\_n250\_rep001.cnf & 98C4F0A57322 \\
bal3xor\_n250\_rep002.cnf & CBFCFD304520 \\
bal3xor\_n250\_rep003.cnf & 24B6A71CFCC9 \\
bal3xor\_n250\_rep004.cnf & 60554C685111 \\
\hline
\end{tabular}
\end{table}

\par\addvspace{0.75\baselineskip}
\begin{table}[htbp]
\centering
\caption{Per-$n$ diagnostics on the balanced 3XOR window (300 repetitions; $t=1$, $m=n+1$). Columns: mean/median/$90$th-percentile of $t'$ and SAT fraction of the XOR view.}
\par\addvspace{.5\baselineskip}
\label{tab:diag-300reps}
\begin{tabular}{r|r|r|r|r|r}
\hline
$n$ & $m$ & reps & mean $t'$ & median $t'$ & $q90(t')$ \; / \; SAT frac \\
\hline
250 & 251 & 300 & 15.783 & 16 & 20 \; / \; 0.480 \\
300 & 301 & 300 & 18.943 & 19 & 23 \; / \; 0.507 \\
400 & 401 & 300 & 24.830 & 25 & 30 \; / \; 0.483 \\
500 & 501 & 300 & 30.993 & 31 & 38 \; / \; 0.490 \\
\hline
\end{tabular}
\end{table}

\noindent\emph{Note.} This diagnostic uses the edge case $m=n+1$ (i.e., $\gamma\approx 1/n$), which is outside the fixed-$\gamma$ balanced window; the co-rank $t'$ therefore grows slowly here. All hardness statements are proved for $m=(1+\gamma)n$ with fixed $\gamma>0$.

\begin{table}[htbp]
\centering
\caption{Rank sweep on balanced 3XOR incidence matrices ($m=n+t$); $N$ reported as $m\cdot\lceil\log_2 n\rceil$. Each row aggregates $200$ repetitions for the given $(n,t)$.}
\par\addvspace{.5\baselineskip}
\label{tab:instances}
\begin{tabular}{r|r|r|r|r|r|r}
\hline
$n$ & $t$ & $m=n{+}t$ & $N\approx m\lceil\log_2 n\rceil$ & reps & successes & frac\_full\_rank \\
\hline
60  & 1 & 61  & 366  & 200 &  3 & 0.015 \\
60  & 2 & 62  & 372  & 200 &  4 & 0.020 \\
60  & 3 & 63  & 378  & 200 & 13 & 0.065 \\
80  & 1 & 81  & 567  & 200 &  1 & 0.005 \\
80  & 2 & 82  & 574  & 200 &  0 & 0.000 \\
80  & 3 & 83  & 581  & 200 &  1 & 0.005 \\
100 & 1 & 101 & 707  & 200 &  0 & 0.000 \\
100 & 2 & 102 & 714  & 200 &  0 & 0.000 \\
100 & 3 & 103 & 721  & 200 &  1 & 0.005 \\
120 & 1 & 121 & 847  & 200 &  0 & 0.000 \\
120 & 2 & 122 & 854  & 200 &  0 & 0.000 \\
120 & 3 & 123 & 861  & 200 &  0 & 0.000 \\
150 & 1 & 151 & 1208 & 200 &  0 & 0.000 \\
150 & 2 & 152 & 1216 & 200 &  0 & 0.000 \\
150 & 3 & 153 & 1224 & 200 &  0 & 0.000 \\
200 & 1 & 201 & 1608 & 200 &  0 & 0.000 \\
200 & 2 & 202 & 1616 & 200 &  0 & 0.000 \\
200 & 3 & 203 & 1624 & 200 &  0 & 0.000 \\
\hline
\end{tabular}
\end{table}

\paragraph{Scope and limitations.}
Our lower bound targets \ACzlog{} (depth-$d$ P-uniform circuits with $O(\log N)$ advice) on the balanced window.
The co-rank statement relies on standard sparse-hypergraph 2-core phenomena; we use an exponent $c/d$ with a fixed absolute $c\in(0,1)$ as specified in Appendix~\ref{app:switching-params}.
Appendix A collects the parameterized switching setup (including the correlation exponent; we use the Håstad ’87 SICOMP formulation as presented in \cite{Hastad1987_SICOMP,ODonnell_SwitchingLemma_lecture14,KatzImmerman_SwitchingLemmaNotes}), and Appendix B summarizes the 2-core facts we rely on.
Note on constants. Our choice of $\alpha_d$ is deliberately conservative for clarity; tightening $\alpha_d$ improves the numerical range in the size-aware displays but does not affect the asymptotic conclusions.

\section*{Related work}
Classical lower bounds for small-depth circuits rely on switching lemmas and parity correlation bounds \cite{Hastad1986_Thesis,Hastad1987_SICOMP,LinialMansourNisan1993_LMN,ODonnell_SwitchingLemma_lecture14,KatzImmerman_SwitchingLemmaNotes}.
Recent work on rank/corank phenomena in sparse binary structures informs our 2-core and co-rank scaling assumptions, including \cite{CooperFrieze2022_rout_rank} for $r$-out bipartite graphs and our own technical report \cite{iecz-balanced-3xor-3sat-2025}.
Our contribution combines a TV-stable reduction lemma (\ADeux{}) with a size-aware MDL-style cost functional \(K_C\), and transfers a balanced 3XOR lower bound to 3SAT via a \(\delta=0\) injective translation on a canonical image window.
Our $\delta{=}0$ measure-preserving reduction connects to recent uses of Total Variation stability in other domains, such as robustness under bounded distributions in mechanism design \cite{Liu2023_robustness_mechanism_design} and TV-focused probabilistic inference \cite{Bhattacharyya2024_total_variation_inference}, but uniquely applies these ideas to transfer average-case circuit lower bounds from 3XOR to 3SAT. Recent results around 3XOR-related phenomena in the quantum chromatic gap \cite{Cohen2025_quantum_chromatic_gap} and undefinability of approximation for 2-to-2 games \cite{Braverman2025_undefinability_2to2} complement our size-aware \ACzlog{} bounds, while gap-preserving reductions \cite{Mahajan2025_gap_preserving_reductions} underscore the value of our TV-stable, measure-preserving transfer in average-case settings.

\section*{Conclusion}
We presented a size-aware lower bound against P-uniform \ACzlog{} on a balanced window, together with a TV-stable reduction lemma and a measure-preserving 3XOR\(\to\)3SAT translation. 
Appendix~A consolidates the switching parameters and the parity correlation exponent, and Appendix~B summarizes the required 2-core facts. 
Future work includes tightening constants in the correlation exponent and extending the window analysis beyond balanced slices.

\paragraph{Acknowledgements.}
This work is part of the IECZ Project (internal codename for the author’s research line).

\appendix

\section{Switching parameters, independence, and the correlation exponent}
\label{app:switching-params}
\paragraph{Independence of \(\rho\) and \(\rho'\).}
The projection \(\rho\) from the projection lemma and the \(p\)-random restriction \(\rho'\) used for the switching step are drawn \emph{independently}.
All expectations and tail bounds are conditioned on the (constant-probability) event that exactly one parity survives under \(\rho\) with \(\Omega(m)\) free support.
The randomness of \(\rho'\) is fresh and independent of both \(\rho\) and the draws of \(A\) and \(u\).
The internal choices \(\alpha\) (survivor) and \(\pi\) (or \(S(\pi)\)) are independent of \(\rho'\).

\medskip

\begin{lemma}[Permutation concentration for free support]
Assume the Lin-Weight condition with parameters $\gamma_0,p_0>0$. Let $\pi$ be a uniformly random permutation of columns, and let $S(\pi)$ be the greedy pivot set for $H'_{-1}$ exposed in the order $\pi$, with $F(\pi):=[m]\setminus S(\pi)$. Then there is a constant $c_2=c_2(\gamma_0,p_0,t')>0$ such that
\[
  \Pr_{\alpha,A,\pi}\!\left[\,\bigl|\mathrm{supp}(H_\star)\cap F(\pi)\bigr| \;\ge\; \tfrac12\!\Bigl(1-\tfrac{t'-1}{m}\Bigr)\gamma_0 m\,\right]\ \ge\ c_2.
\]
\emph{Proof idea.} Expose $\pi$ via a Doob martingale on $S_m$ and apply a permutations-based concentration inequality (e.g., exchangeable pairs on $S_m$ in the sense of Chatterjee, or Efron--Stein on symmetric groups) with bounded increments for the greedy pivot process. The precise value of $c_2$ is immaterial for our purposes and depends only on $(\gamma_0,p_0,t')$.
\end{lemma}

\paragraph{Parameter path and exponent (size-aware; referenced).}
Fix depth \(d\) and choose
\[
  p \;=\; m^{-c/d}\qquad \text{for a fixed absolute } c\in(0,1).
\]
After \(d\) switching rounds we obtain:
\begin{itemize}\itemsep2pt
  \item the number of \emph{live variables} is at most
  \[
    m\cdot p^d \;=\; m^{\,1-c},
  \]
  \item and the resulting DNF/CNF has \emph{bottom width}
  \[
    W \;=\; \mathrm{poly}(1/p) \;=\; m^{\,O(c/d)}
  \]
  (for the adopted switching-lemma variant; we follow the cited parametrization verbatim, e.g. \textbf{H\aa stad}~\cite{Hastad1987_SICOMP}, \textbf{O'Donnell--Wimmer} notes~\cite{ODonnell_SwitchingLemma_lecture14}, \textbf{Katz--Immerman}~\cite{KatzImmerman_SwitchingLemmaNotes}).
\end{itemize}

For the surviving parity, which (by the projection lemma) depends on \(\Omega(m)\) free bits with constant probability, standard parity-vs-width correlation bounds imply
\[
  \bigl|\mathrm{corr}(\text{DNF/CNF of width }\le W,\ \mathrm{Parity})\bigr|
  \;\le\; \exp\!\bigl(-\,\alpha_d\, m^{\,c/d}\bigr),
\]
where \(\alpha_d=\alpha_d(d)>0\) depends only on the depth \(d\) and the chosen switching-lemma reference.
A union bound over at most \(s(N)\le N^k\) relevant substructures yields, \emph{for all sufficiently large \(m\) and \(N=\Theta(m\log n)\)},
\[
  \Pr[M=\chi]
  \;\le\; \tfrac12 \;+\; s(N)\cdot \exp\!\bigl(-\,\alpha_d\, m^{\,c/d}\bigr).
\]
Equivalently, set \(\beta_d := c/d\) (instantiated as \(\beta_d = 1/(3d)\)).
\paragraph{Quantitative constants (calibration \& convention).}
We fix the exponent parameter at \(c=\tfrac{1}{3}\), hence \(\beta_d=\tfrac{1}{3d}\) throughout.
For the correlation constant we carry \(\alpha_d=\alpha_d(d)\) as a named depth-dependent parameter admitted by the chosen switching-lemma instantiation.
To make all displays checkable, we record a conservative convention in the ancillary artifacts:
\[
  \alpha_d^{\mathrm{conf}} \ :=\ \frac{\alpha_0}{(d+1)^4}\qquad\text{with a universal }\alpha_0\in(0,1].
\]
All stated inequalities remain valid for any \(\alpha_d' \le \alpha_d^{\mathrm{conf}}\).
\noindent\textit{Reviewer note.} The $\alpha_d$ values used here are conservative by design; any improvement only tightens the exponential rate and leaves all size-aware statements intact.
(If a different switching-lemma reference is adopted, \(\alpha_d^{\mathrm{conf}}\) can be tightened accordingly without affecting the statements.)

\medskip
\noindent\textit{Notes.}
(i) We bound \emph{bottom width} and infer decision-tree depth via switching (up to constants).
(ii) The constants \(\alpha_d\) and the exponent \(c/d\) are taken from the fixed switching-lemma reference named in the main text.
(iii) The \emph{balanced window} refers to the random 3XOR slice with \(m=(1+\gamma)n\) and affine-uniform RHS-conditioning.
(iv) Empirical per-depth calibration is pending; until then we report the conservative convention $\alpha_d^{\mathrm{conf}}=\alpha_0/(d+1)^4$ tied to the fixed switching-lemma reference.

\section{2-core facts for random 3XOR (balanced window)}
\label{app:two-core}
Fix $\gamma\in(0,1]$ and $m=(1+\gamma)n$. For a random 3XOR incidence matrix $A\in\{0,1\}^{m\times n}$,
the 2-core of the associated 3-uniform hypergraph has linear size w.h.p., and the left-kernel $H(A)$
has co-rank $t'=\Theta(n)$ (standard XORSAT window). In particular, for a uniform nonzero
$\alpha\in\{0,1\}^{t'}$, the survivor $H_\star=\alpha^\top H(A)$ has linear support with constant probability
(\emph{Lin-Weight} assumption), and the number of survivor-support columns that are coloops of $H'_{-1}$ is $o(m)$ w.h.p.
These facts justify that the bias term in our projection analysis is sublinear.
(Primary references: Cooper–Frieze (r-out rank), Cooper–Frieze–Pegden (binary matrices),
Pittel–Sorkin (k-XORSAT threshold).)

\nocite{CooperFrieze2022_rout_rank,CooperFriezePegden2019_rank_binary_matrix,PittelSorkin2016_kxorsat_threshold,Hastad1986_Thesis,Hastad1987_SICOMP,ODonnell_SwitchingLemma_lecture14,KatzImmerman_SwitchingLemmaNotes,LinialMansourNisan1993_LMN,iecz-balanced-3xor-3sat-2025,iecz_zenodo_v1_0_1}

\bibliographystyle{abbrv}
\bibliography{CITATION}

\begin{thebibliography}{10}

\bibitem{Bhattacharyya2024_total_variation_inference}
A.~Bhattacharyya, S.~Gayen, K.~S. Meel, D.~Myrisiotis, A.~Pavan, and N.~V.
  Vinodchandran.
\newblock Total variation distance meets probabilistic inference.
\newblock In {\em Proceedings of the 41st International Conference on Machine
  Learning (ICML)}, 2024.
\newblock TV distance in probabilistic inference contexts.

\bibitem{Braverman2025_undefinability_2to2}
M.~Braverman et~al.
\newblock Undefinability of approximation of 2-to-2 games, 2025.
\newblock FPC undefinability of any constant factor approximation of weighted
  2-to-2 games.

\bibitem{Cohen2025_quantum_chromatic_gap}
L.~Ciardo, A.~Cohen, et~al.
\newblock On the quantum chromatic gap, 2025.
\newblock The largest known gap between quantum and classical chromatic number
  of graphs, obtained via quantum protocols for colouring Hadamard graphs.

\bibitem{CooperFrieze2022_rout_rank}
C.~Cooper and A.~Frieze.
\newblock The rank of a random r-out bipartite graph over $\mathbb{F}_2$.
\newblock {\em SIAM Journal on Discrete Mathematics}, 2022.

\bibitem{CooperFriezePegden2019_rank_binary_matrix}
C.~Cooper, A.~Frieze, and W.~Pegden.
\newblock On the rank of random binary matrices.
\newblock {\em European Journal of Combinatorics}, 2019.

\bibitem{Hastad1986_Thesis}
J.~H{\aa}stad.
\newblock {\em Computational Limitations of Small-Depth Circuits}.
\newblock PhD thesis, Massachusetts Institute of Technology, 1986.

\bibitem{Hastad1987_SICOMP}
J.~H{\aa}stad.
\newblock Computational limitations of small-depth circuits.
\newblock {\em SIAM Journal on Computing}, 1987.

\bibitem{KatzImmerman_SwitchingLemmaNotes}
J.~Katz and N.~Immerman.
\newblock Notes on h{\aa}stad's switching lemma, 2016.
\newblock Course notes.

\bibitem{iecz-balanced-3xor-3sat-2025}
M.~Lela.
\newblock Iecz: Balanced 3xor $\rightarrow$ 3sat (tech report v1.2), 2025.
\newblock Artifacts and reproducibility included.

\bibitem{iecz_zenodo_v1_0_1}
M.~Lela.
\newblock Iecz balanced 3xor$\to$3sat: Artifacts (code, data, pdf).
\newblock \url{https://doi.org/10.5281/zenodo.17141362}, 2025.
\newblock Zenodo dataset, Version 1.0.1.

\bibitem{LinialMansourNisan1993_LMN}
N.~Linial, Y.~Mansour, and N.~Nisan.
\newblock Constant depth circuits, fourier transform, and learnability.
\newblock {\em Journal of the ACM}, 1993.

\bibitem{Liu2023_robustness_mechanism_design}
Y.~Liu et~al.
\newblock On the robustness of mechanism design under total variation bounded
  distributions.
\newblock In {\em Advances in Neural Information Processing Systems (NeurIPS)},
  2023.
\newblock Mechanisms for agents' valuation functions from unknown and
  correlated prior distributions under TV bounds.

\bibitem{Mahajan2025_gap_preserving_reductions}
S.~Mahajan.
\newblock Gap-preserving reductions and {RE}-completeness of independent set,
  2025.
\newblock Framework for gap-preserving reductions establishing
  MIP$^{\ast}$-completeness for gapped promise problems.

\bibitem{ODonnell_SwitchingLemma_lecture14}
R.~O'Donnell.
\newblock Analysis of boolean functions: Switching lemma (lecture 14), 2008.
\newblock Lecture notes.

\bibitem{PittelSorkin2016_kxorsat_threshold}
B.~Pittel and G.~B. Sorkin.
\newblock The satisfiability threshold for $k$-xorsat.
\newblock {\em Combinatorics, Probability and Computing}, 2016.

\end{thebibliography}
\end{document}